\documentclass[12pt]{article}
\usepackage[utf8]{inputenc}
\usepackage[english]{babel}
\usepackage{amsfonts,amsthm}
\usepackage{amsmath,amssymb}
\usepackage{xcolor}
\usepackage{tikz}
\usepackage{cite}

\newcommand{\be}{\begin{equation}}
\newcommand{\ee}{\end{equation}}
\newcommand{\bea}{\begin{eqnarray}}
\newcommand{\eea}{\end{eqnarray}}

\renewcommand{\Re}{\mathrm{Re }}

\newcommand{\triplet}[3]{ \left( \begin{array}{c}#1 \\ #2 \\ #3 \end{array}\right) }
\newcommand{\mmatrix}[4]{\left(\!\!\begin{array}{cc}#1 & #2 \\ #3 & #4 \end{array}\!\!\right)}
\newcommand{\mmmatrix}[9]{ \left(\! \begin{array}{ccc}#1 & #2 & #3\\ #4 & #5 & #6\\ #7 & #8 & #9\\ \end{array}\!\right) }

\newcommand{\Tr}{\mathrm{Tr}}
\newcommand{\Z}{\mathbb{Z}}
\newcommand{\RR}{\mathbb{R}}
\newcommand{\CC}{\mathbb{C}}

\newcommand{\id}{\mathbf{1}}

\newtheorem{theorem}{Theorem}
\newtheorem{proposition}{Proposition}
\newtheorem{lemma}{Lemma}

\def\lsim{\mathrel{\rlap{\lower4pt\hbox{\hskip1pt$\sim$}}
    \raise1pt\hbox{$<$}}}         %less than or approx. symbol
\def\gsim{\mathrel{\rlap{\lower4pt\hbox{\hskip1pt$\sim$}}
    \raise1pt\hbox{$>$}}}         %greater than or approx. symbol

\title{Bounded-from-below conditions for $A_4$-symmetric 3HDM}

\author{N.~Buskin$^{1}$\thanks{E-mail: nvbuskin@gmail.com}, Igor~P.~Ivanov$^{2}$\thanks{E-mail: ivanov@mail.sysu.edu.cn}  
\\
  {\small $^1$ Dept. of Mechanics and Mathematics, Novosibirsk State University, 630090 Novosibirsk, Russia}\\
  {\small $^2$ School of Physics and Astronomy, Sun Yat-sen University, 519082 Zhuhai, China}
  }

\begin{document}

\maketitle

\bigskip
\begin{abstract}
Deriving necessary and sufficient conditions for a scalar  potential to be bounded from below (BFB)
is a difficult task beyond the simplest cases.
Recently, a set of BFB conditions was proposed for the $A_4$-invariant three-Higgs-doublet model (3HDM).
However, that set of conditions relied on numerical scan, and a complete analytic proof was lacking.
Here, we fill this gap. We prove that the conjectured BFB conditions are indeed necessary and sufficient
within the neutral Higgs subspace.
We bypass technically challening direct algebraic computations with a novel technique
that relies on an auxiliary function, which is related to the Higgs potential but which is easier to analyze.
This technique may finally be sufficient to tackle the more involved case of the original Weinberg's 3HDM model.
\end{abstract}

\tableofcontents

\section{Introduction}

\subsection{Physics motivation} 

The LHC discovery of the 125 GeV scalar boson \cite{Aad:2012gk,Chatrchyan:2012gu,Khachatryan:2016vau} completed the Standard Model 
and proved that the Brout-Englert-Higgs mechanism \cite{Englert:1964et,Higgs:1964ia,Higgs:1964pj,Guralnik:1964eu} is indeed at work.
However the Higgs sector does not need to be as minimal as postulated in the Standard Model (SM).
The simple idea that the Higgs doublets come in generations, just like fermions,
offers many opportunities to address the several shortcomings of the SM.
This class of models, known as the $N$-Higgs-doublet models (NHDM), first appeared in 1970's \cite{Lee:1973iz,Weinberg:1976hu},
and many of its versions were studied in thousands of papers since then, 
see e.g. reviews \cite{Branco:2011iw,Ivanov:2017dad} and references therein.

Before a phenomenological analysis of a multi-Higgs model can begin, 
one must ensure the mathematical self-consistency of the model. 
One of the requirements is that the Higgs potential
be bounded from below (BFB), at least at the tree level, in order for the vacuum state to exist. 
This implies that the coefficients of the quartic part of the potential
must satisfy certain inequalities known as the BFB conditions.
In simple cases, these can be written directly \cite{Branco:2011iw}; 
in other cases, one needs to resort to more elaborate mathematical tools.

In the general 2HDM, the necessary and sufficient BFB conditions are known \cite{Ivanov:2006yq}.
For three Higgs doublets, the problem remains unsolved in its generality.
However, in particular versions of the 3HDM equipped with additional global symmetries,
the structure of the Higgs potential simplifies, 
and in certain cases the necessary and sufficient BFB conditions are also known.
These examples include the 3HDM with the $U(1)\times U(1)$ symmetry group \cite{Faro:2019vcd}, 
where the copositivity methods are at work \cite{Kannike:2012pe,Kannike:2016fmd},
and the $S_4$-symmetric 3HDM \cite{Ivanov:2020jra}.

The situation with the exact BFB conditions for the $A_4$ symmetric 3HDM, a rather popular model 
\cite{Ma:2001dn,Lavoura:2007dw,Morisi:2009sc,Machado:2010uc,Ishimori:2010au,Toorop:2010ex,Boucenna:2011tj,Ivanov:2012fp,Pramanick:2017wry,Chakrabarty:2018yoy},
remains unsettled.
On the one hand, it is very similar to $S_4$ 3HDM, the only difference being that 
one of the coefficient is allowed to be complex.
However, this extra parameter renders the task of establishing the exact BFB conditions surprisingly difficult. 
Several previous phenomenological papers used conflicting conditions,
which, as was pointed out in the recent work \cite{Ivanov:2020jra}, were either incomplete or incorrect.
The same paper \cite{Ivanov:2020jra}, applying the methods of \cite{Degee:2012sk}, 
proposed the full set of the necessary and sufficient BFB conditions 
under the simplifying assumption that the vacuum state is neutral, not charge breaking.
However, the derivation relied on certain results of \cite{Degee:2012sk} which had only been demonstrated
by a numerical scan, not analytically.
Although these BFB conditions were confirmed in all the numerical examples tested,
a direct analytical proof of their validity was still lacking.

The purpose of this paper is to fill this gap.
We prove a theorem which confirms that the conditions conjectured in \cite{Ivanov:2020jra}
indeed represent the necessary and sufficient BFB conditions 
for $A_4$ 3HDM with neutral vacuum. We also develop a novel approach which bypasses the technical challenges of 
straightforward calculations by introducing an auxiliary function,
which is related to the Higgs potential but which is easier to analyze.
We believe this method could be eventually applicable to other popular 3HDMs, 
in particular to Weinberg's \cite{Weinberg:1976hu} and Branco's \cite{Branco:1980sz} models.

Since the task is purely mathematical,
we reformulate the problem in notation which is convenient for mathematical analysis.
We will switch back to the physics language in the last section of the paper.
An appendix provides further auxiliary mathematical details referred in the main text.

\subsection{Mathematical formulation}

Let $\vec z = (z_1, z_2, z_3)$ be a vector in $\CC^3$;
$\bar z_j$ is the complex conjugation of $z_j$.
Consider the following real-valued polynomial (the quartic Higgs potential):
\begin{eqnarray}
V_4 &=& a\left(\bar{z}_{1}z_1 + \bar{z}_{2}z_2 + \bar{z}_{3}z_3\right)
+ 2b\left(|z_1z_2| + |z_2z_3| + |z_3z_1|\right)\nonumber\\
&+& d(\bar{z}_{1}z_2 + \bar{z}_{2}z_3 + \bar{z}_{3}z_1) 
+ \bar{d}(\bar{z}_{2}z_1 + \bar{z}_{3}z_2 + \bar{z}_{1}z_3)\,,\label{V4-1} 
\end{eqnarray}
where $a$ and $b$ are real, while $d$ can be complex.
Our task is to find the necessary and sufficient conditions on the parameters $a, b, d$
which guarantee that $V_4 \ge 0$ everywhere in $\CC^3$ (the BFB conditions). 

For real $d$, the full set of BFB conditions was recently established \cite{Ivanov:2020jra,Degee:2012sk}:
\begin{equation}
a \ge 0\,, \quad a+2b+2d \ge 0\,, \quad a+2b-d \ge 0\,, \quad a+b-d \ge 0\,.\label{S4-BFB}
\end{equation}
If $d < 0$, the last two conditions are weaker and can be dropped. If $d>0$, the second condition can be dropped, instead.
For a complex $d=|d|e^{i\gamma}$, the BFB conditions were \textit{conjectured} in the same paper \cite{Ivanov:2020jra} to be
\begin{eqnarray}
&&a \ge 0\,, \quad a+b-|d| \ge 0\,,\label{A4-BFB-1}\\[2mm]
&&a+2b+2|d|\cos\left(\gamma +\frac{2\pi k}{3}\right) \ge 0\,, \quad k = 1,2,3\,.\label{A4-BFB-2}
\end{eqnarray}
One can verify that for $\gamma=0$ or $\pi$ one obtains a set of conditions that is equivalent to Eqs.~\eqref{S4-BFB}.
The goal of the present paper is to prove the following theorem.
%%%%%%%%%%%%%%%%%%%%%%%%%%%%%%%%%%%%%%%%%%%%
\begin{theorem}[Main theorem]
\label{Main-theorem}
The quartic potential $V_4$ defined in Eq.~\eqref{V4-1} is non-negative in the entire $\CC^3$ if and only if the conditions
\eqref{A4-BFB-1} and \eqref{A4-BFB-2} are satisfied.
\end{theorem}
%%%%%%%%%%%%%%%%%%%%%%%%%%%%%%%%%%%%%%%%%%%%
Before we go into details, let us comment that we have verified this %conclusions 
statement numerically via a scan in the parameter space \cite{MCAworkshop}. However proving it analytically
turned out surprisingly difficult. One reason is that direct differentiation produced a system of coupled 
algebraic and trigonometric equations, which are hard to manipulate.
We also tried various geometric methods, but they gave only partial results.
We could obtain a complete proof only with a rather indirect method to be explained below.
However we admit that a simpler, more direct solution to the same problem may exist.

\section{The main theorem}

\subsection{Preliminary remarks}
Let us first rewrite the quartic potential \eqref{V4-1} as
\begin{eqnarray}
V_4 &=& \tilde a\left(\bar{z}_{1}z_1 + \bar{z}_{2}z_2 + \bar{z}_{3}z_3\right)
+ b\left(|z_1| + |z_2| + |z_3|\right)^2\nonumber\\
&+& d(\bar{z}_{1}z_2 + \bar{z}_{2}z_3 + \bar{z}_{3}z_1) 
+ \bar{d}(\bar{z}_{2}z_1 + \bar{z}_{3}z_2 + \bar{z}_{1}z_3)\,,\label{V-abcd-rank1-alg2} 
\end{eqnarray}
where $\tilde a = a-b$.
It can be written as a sum of a hermitean quadratic form in $\CC^3$ and an additional part:
\begin{equation}
V_4 = h(\vec z) + b f(\vec z)^2\,,\label{V4-3}
\end{equation}
where 
\begin{equation}
h(\vec z) = \vec z\,{}^\dagger A \vec z\,, \quad A = \mmmatrix{\tilde a}{d}{\bar d}{\bar d}{\tilde a}{d}{d}{\bar d}{\tilde a} \label{matrix-A}
\end{equation}
and 
\begin{equation}
f(\vec z) = |z_1|+|z_2|+|z_3|\,.\label{f}
\end{equation}
If the additional part were absent, $b=0$, the BFB conditions could be obtained in a straightforward way.
By direct computation, one finds the spectrum of the hermitean matrix $A$:
\begin{equation}
\lambda_{1,2,3} = \tilde a + 2|d| \cos\left(\gamma + \frac{2\pi k}{3}\right) = 
\tilde a + 2 \Re\left(d\, \omega^k\right)\,,\quad k = 1, 2, 3.\label{spectrum}
\end{equation}
Here, $\omega$ denotes the cubic root of one: $\omega = \exp(2\pi i/3)$. The corresponding normalized eigenvectors 
are
\begin{equation}
n_1 = \frac{1}{\sqrt{3}}\triplet{1}{\omega}{\omega^2}\,,\quad
n_2 = \frac{1}{\sqrt{3}}\triplet{1}{\omega^2}{\omega}\,,\quad
n_3 = \frac{1}{\sqrt{3}}\triplet{1}{1}{1}\,.\label{eigenvectors}
\end{equation}
In this case, non-negativity of $V_4$ in the entire $\CC^3$ is equivalent
to the requirement that all three eigenvalues be non-negative, $\lambda_k \ge 0$.
For $b=0$, these conditions coincide with \eqref{A4-BFB-2}, while the remaining conditions \eqref{A4-BFB-1} become redundant.

For a non-zero \textit{negative} $b$, the BFB conditions could also be found directly.
Since $V_4$ is a homogeneous function, checking that $V_4 \ge 0$ on the unit sphere $|z_1|^2+|z_2|^2+|z_3|^2 = 1$
is equivalent to checking $V_4 \ge 0$ in the entire $\CC^3$. 
In order to verify $V_4 \ge 0$ on the unit sphere, we can try to minimize, separately, the hermitean part $h(\vec z)$ 
and the additional piece $-|b| f^2$, that is, to maximize $f^2$. 
Let us rewrite $f^2$ on the unit sphere as
\begin{eqnarray}
(|z_1|+|z_2|+|z_3|)^2 &=& 3 (|z_1|^2+|z_2|^2+|z_3|^2) \nonumber\\
&&\ \ - \left(|z_1|-|z_2|\right)^2 - \left(|z_2|-|z_3|\right)^2 - \left(|z_3|-|z_1|\right)^2\nonumber\\
&=& 3 - \left(|z_1|-|z_2|\right)^2 - \left(|z_2|-|z_3|\right)^2 - \left(|z_3|-|z_1|\right)^2\,.
\end{eqnarray}
This function is maximized when $|z_1|=|z_2|=|z_3|$. 
Therefore, not only does one of the eigenvectors \eqref{eigenvectors} minimize $h(\vec z)$, but also it maximizes $f^2$.
We conclude that, for $b < 0$, the necessary and sufficient BFB conditions are
\begin{equation}
\lambda_k + 3b = a + 2b + 2|d| \cos\left(\gamma + \frac{2\pi k}{3}\right) \ge 0\,.
\end{equation}
These conditions coincide with \eqref{A4-BFB-2}, while \eqref{A4-BFB-1} remain redundant.

The analysis becomes complicated only when $b > 0$. Since we could not find an equally elegant shortcut to the answer,
we establish below an indirect method which proves that, even in this case, conditions
\eqref{A4-BFB-1} and \eqref{A4-BFB-2} give the necessary and sufficient BFB conditions. 

\subsection{The restricted problem and barycenters}

Let us restrict the analysis of $V_4$ to the following section:
\begin{equation}
{\mathcal T} =\{(z_1,z_2,z_3)\in \mathbb{C}\,|\,|z_1|+|z_2|+|z_3|=1\}.
\end{equation}
Clearly, $V_4|_{\mathcal T}=h(\vec{z})+ b$.
Every ray in $\CC^3$ starting from the origin intersects ${\mathcal T}$.
Since $V_4$ is a a homogeneous function, the condition that $V_4|_{\mathcal T} \ge 0$ is equivalent to 
$V_4\ge 0$ in the entire $\CC^3$.

Next, consider the following triangular region in $\RR^3$:
\begin{equation}
%T=\{(x,y,z)\,|\,x+y+z=1, x\geqslant 0,y\geqslant 0,z\geqslant 0\}.
T=\{(x,y,z)\,|\,x+y+z=1,\  x\ge 0,\ y\ge 0,\ z\ge 0\}.
\end{equation}
We get a natural projection map $p\colon {\mathcal T} \to T$ defined by
$(z_1,z_2,z_3)\mapsto (|z_1|,|z_2|,|z_3|)$.
This projection map allows us to view ${\mathcal T}$ as a fiber bundle over the base space $T$
with fibers being the 3-tori $U(1)\times U(1)\times U(1)$ over the internal points of $T$,
2-tori $U(1)\times U(1)$ over the internal points of the sides of $T$, and 
%1-tori 
circles $U(1)$ over the vertices.

Let us now consider the set of \textit{barycenters} of $T$: 
\begin{eqnarray}
\mbox{the internal barycenter:}&& (1/3,\, 1/3,\, 1/3)\,,\nonumber\\
\mbox{the side barycenters:}&& (1/2,\, 1/2,\, 0)\,,\quad (1/2,\, 0,\, 1/2)\,,\quad (0,\, 1/2,\, 1/2)\,,\label{barycenters}\\
\mbox{the vertices:}&& (1,\, 0,\, 0)\,,\quad(0,\, 1,\, 0)\,,\quad (0,\, 0,\, 1)\,.\nonumber
\end{eqnarray}
The following proposition clarifies the role of barycenters in the BFB problem.

\begin{proposition}
The conditions \eqref{A4-BFB-1}, \eqref{A4-BFB-2} hold if and only if $V_4 \ge 0$ 
in the fibers above all the barycenters \eqref{barycenters}.
\label{barycenter-proposition}
\end{proposition}

\begin{proof}
Suppose the conditions \eqref{A4-BFB-2} are satisfied.
For any $\vec{z}$ satisfying $||\vec{z}||^2 = |z_1|^2+|z_2|^2+|z_3|^2=1/3$ we get
\begin{equation} 
h(\vec{z})\geqslant ||\vec{z}||^2 \cdot\underset{k=1,2,3}{\min} \,\lambda_k = \frac{1}{3}\ \underset{k=1,2,3}{\min} \,\lambda_k\,.
\end{equation}
Therefore, the function $h(\vec{z})+b \ge 0$ everywhere in the sphere 
$||\vec{z}||^2 =1/3$, 
including the space of all vectors $\vec{z}$ such that $|z_1|=|z_2|=|z_3|=1/3$,
which is the fiber above the internal barycenter: $p^{-1}(1/3, 1/3, 1/3)$. 
Within this subspace, $h(\vec{z})+ b$ coincides with $V_4$. Therefore, %$V_4 \le 0$ 
$V_4 \ge 0$ in the fiber above the internal barycenter.
 
Conversely, suppose 
%$V_4 \le 0$ 
$V_4 \ge 0$ everywhere in $p^{-1}(1/3, 1/3, 1/3)$.
Then, computing it over the vectors $\vec n/\sqrt{3}$, where $\vec n$ are given by \eqref{eigenvectors},
we obtain exactly the same expressions as in \eqref{A4-BFB-2} divided by 3. Therefore, they are non-negative.

Next, suppose the second condition of \eqref{A4-BFB-1} holds. 
A generic point in the fiber over the barycenter $(1/2, 1/2, 0)$ can be parametrized as
$(e^{i\alpha}/2,\, e^{i\beta}/2,\, 0)$.
Computing $V_4$ on this fiber, we get
\begin{equation}
V_4 = \frac{1}{2}\tilde a + \frac{1}{2}|d|\cos(\gamma+\beta-\alpha) + b \ge \frac{1}{2}\left(a+b-|d|\right)\ge 0\,.
\end{equation}
The same holds for the other two side barycenters.
Conversely, if $V_4 \ge 0$ everywhere in $p^{-1}(1/2,1/2,0)$, then one observes that 
its minimal value corresponds to the second condition of \eqref{A4-BFB-1}.

Finally, since $a = \tilde a + b$, it coincides with the value of $V_4$
in the barycenters above the vertices such as $p^{-1}(1,0,0)$. 
The proof is complete.
\end{proof}
This proposition allows us to recast the main theorem in the following form:
\begin{theorem}[Main theorem via barycenters]
The quartic potential $V_4 \ge 0$ in $\CC^3$ if and only if $V_4 \ge 0$ in the fibers above all the barycenters \eqref{barycenters}.
\end{theorem}

\subsection{The idea of the proof}\label{subsection-idea}

Let us now describe the main idea of the proof.
A straightforward analysis of $V_4$ in $\CC^3$ is difficult because of large dimensionality of the problem.
Even when restricted to ${\mathcal T}$, $V_4$ depends on four independent variables: two absolute values and two 
phases of $z_1$, $z_2$, $z_3$. Searching for its constrained minimum leads to a coupled system of algebraic and 
trigonometric equations, which are hard to solve.

We will avoid this problem by defining an auxiliary function $\varphi$, which depends only on the absolute values of $z_i$, 
not on their phases. This function has the important property that 
$V_4(\vec{z}) \geqslant \varphi(|z_1|,|z_2|,|z_3|) + b$ for all points in $\CC^3$ where $V_4$ has an extremum on ${\mathcal T}$.
This applies also to the point where $V_4$ has the global minimum on ${\mathcal T}$.
As a result, for any $\vec{z}\in \mathcal T$, the following chain of inequalities is true:
\begin{equation}
V_4(\vec{z})\geqslant \underset{\vec{z} \in \mathcal T}{\min}\,V_4(\vec{z}) \geqslant \underset{(x,y,z)\in T}{\min}\,\varphi(x,y,z)+b\,.\label{chain}
\end{equation}
Thus, instead of investigating the function $V_4$ in ${\mathcal T}$, we will look for the global minimum of $\varphi$ in $T$,
which depends only on two independent variables.

We will find that the global minimum of $\varphi$, if it is non-zero, can only be at the barycenters of $T$. 
In this case, we will show that the inequalities \eqref{A4-BFB-1}, \eqref{A4-BFB-2}
lead to $\underset{(x,y,z)\in T}{\min}\,\varphi(x,y,z)+b \ge 0$ and, therefore, to $V_4 \ge 0$ in the entire $\CC^3$.
If the global minimum of $\varphi$ is zero, we will deduce $V_4(\vec{z}) \ge 0$ directly from the above inequalities. 

We find it important to stress that we compare $V_4$ and $\varphi+b$ globally, not locally.
That is, we are not going to claim that all extremal points of $V_4$ correspond to extremal points of $\varphi$ or vice versa.
We just prove and then use the property that the value of $V_4$ at its global minimum is not smaller 
than the global minimum of $\varphi + b$. Therefore, it suffices to find the global minimum of $\varphi+b$ in $T$
and verify that the conditions \eqref{A4-BFB-1} and \eqref{A4-BFB-2} keep it non-negative.

\section{The function $\varphi$ and its global minima}

\subsection{Defining the auxiliary function $\varphi$}

Let us focus on the interior of $\mathcal T$, where $|z_i| \not = 0$.
In order to find a constrained extremum of $V_4$ on ${\mathcal T}\subset \mathbb{C}^3$ defined by $f^2 = 1$,
we use the method of Lagrange multipliers and differentiate $V_4$ and $f$ with respect to the antiholomorphic coordinates $\bar z_i$.
We remind that, if $z = x+iy$, the antiholomorphic derivative operator is defined as
\begin{equation}
\bar \partial = \frac{\partial}{\partial \bar z}=\frac{1}{2}\left(\frac{\partial}{\partial x}+i \frac{\partial}{\partial y}\right),\quad
\frac{\partial z}{\partial \bar z}=0\,, \quad 
\frac{\partial \bar z}{\partial \bar z}=1\,, \quad  
\frac{\partial |z|}{\partial \bar z}= \frac{1}{2}\frac{z}{|z|}\,.
\end{equation}
The constrained extremum is given by the equation
\begin{equation}
\nabla_{\bar \partial}V_4|_{\mathcal T}=\lambda\nabla_{\bar \partial}f|_{\mathcal T},
\end{equation}
which can be rewritten in terms of $\nabla_{\bar \partial}f^2|_{\mathcal T}$, by rescaling and shifting  the original Lagrange multiplier $\lambda$,  namely
$\lambda=2(\mu +b )$, as
\begin{equation}
\nabla_{\bar \partial}V_4|_{\mathcal T}=(\mu + b)\nabla_{\bar \partial}f^2|_{\mathcal T}.
\end{equation} 
This leads to the equation
\begin{equation}
\label{basic-eigen-equation}
A\triplet{z_1}{z_2}{z_3} = 
\mu \triplet{z_1/|z_1|}{z_2/|z_2|}{z_3/|z_3|}\,,\quad |z_1|+|z_2|+|z_3|=1.
\end{equation}
One can see that, at the points satisfying \eqref{basic-eigen-equation}, 
one has $\mu=\vec{z}^{\,\dagger} A\vec{z}=h(\vec{z}) = V_4(\vec z) - b$. 
Thus, the non-negativity of $V_4$ at these extremum points is equivalent to $\mu + b \ge 0$.

However, the values of $\mu$ just introduced are defined for the points of constrained extrema of $V_4$,
We want to build a function $\varphi$ which is defined everywhere on $T$ and which coincides with $\mu$ at these points.
To build it, let us define the rescaling matrix 
\begin{equation}
R = \begin{pmatrix}
|z_1|^{1/2} & 0 & 0\\
0 & |z_2|^{1/2} & 0\\
0 & 0 & |z_3|^{1/2}
\end{pmatrix}
\end{equation}
and multiply \eqref{basic-eigen-equation} by $R$.
Then Eq.~\eqref{basic-eigen-equation} can be recast in the eigenvalue problem of a new matrix 
$A(x,y,z)=RAR$, where $(x,y,z)=p(\vec{z})$, defined everywhere on $T$:
\begin{equation}
A(x,y,z)\vec w = \mu \vec w\,, \quad
A(x,y,z) = \begin{pmatrix}
x\,\widetilde{a} & \sqrt{xy} \,d & \sqrt{xz}\, \overline{d}\\
\sqrt{xy}\, \overline{d} & y\, \widetilde{a} & \sqrt{yz} \,d\\
\sqrt{xz}\, d & \sqrt{yz}\, \overline{d} & z\, \widetilde{a}
\end{pmatrix}\,,\quad
\vec w = R^{-1} \vec z, 
\end{equation} 
where vector $\vec w=(w_1,w_2,w_3)^t$ is a unit length vector 
satisfying $|w_1|=\sqrt{x},|w_2|=\sqrt{y},|w_3|=\sqrt{z}$. 

The new matrix $A(x,y,z)$ is a hermitean matrix defined everywhere on $T$ including its boundary $\partial T$.
Therefore, one can study its spectrum everywhere in $T$. We define the auxiliary function $\varphi(x,y,z)$
as the minimal eigenvalue (among the three eigenvalues available) at each point in $T$:
\begin{equation}
\label{min-function}
\varphi \colon T\to \mathbb{R},\quad
\varphi(x,y,z)= \min\,Spec\,A(x,y,z).
\end{equation}
Although it is possible to write explicit expressions for the three eigenvalues of a $3\times 3$ matrix,
we will not use them.  
The key observation is that, at the points of constrained extrema of $V_4$,
\begin{equation}
V_4(\vec{z})-b=\mu \geqslant \varphi(|z_1|,|z_2|,|z_3|)\geqslant \underset{(x,y,z) \in T}{\min}\,\varphi(x,y,z)\,.
\end{equation}
We conclude that, for any $\vec{z} \in \mathcal T$,
\begin{equation}
\label{Estimate-inequality}
V_4(\vec{z})\geqslant \underset{\vec{z} \in \mathcal T}{\min}\,V_4(\vec{z}) \geqslant \underset{(x,y,z) \in T}{\min}\,\varphi(x,y,z)+b,
\end{equation} 
which is exactly what we announced in Section~\ref{subsection-idea}. 

With this result, we now turn to the search of the global minimum of function $\varphi(x,y,z)$ instead of dealing with $V_4$.

\subsection{Smooth minima of $\varphi$ inside $T$}
\label{smooth-interior-minima-section}
Analyzing $\varphi$ defined for the most general matrix $A$ and the coefficient $b$ will result in very technical case-dependent analysis. 
Instead we show that it is sufficient to prove
Theorem \ref{Main-theorem} for those functions $V_4$ satisfying \eqref{A4-BFB-1}, \eqref{A4-BFB-2}, which have {\it generic} data
$(A,b)$, see Appendix~\ref{Generic-section} for the relevant definitions and the justification why the problem can be reduced to the generic case.   
So, in the present section we assume that the tuple of coefficients $(A,b)$ determining our $V_4$ is generic with respect to the collection of functions specified in subsection \ref{Collection-subsection}.

At each point in $T$, the matrix $A(x,y,z)$ has three eigenvalues (counted with multiplicity).
Thus, the spectrum of $A(x,y,z)$ defines a triple cover of $T$. Individual branches can intersect,
which means that there may exist points in $T$ where $\varphi(x,y,z)$ is non-differentiable.
Therefore, when looking for the minimum of $\varphi$, we need to analyze not only smooth extrema 
but also the non-differentiable points.

We begin here with smooth critical points lying in the interior of $T$.
Given the matrix $A(x,y,z)$, its eigenvalues $\mu$ can be found as solutions of the characteristic equation: 
$\det(\mu \id - A(x,y,z)) = 0$. The matrix $A(x,y,z)$ depends on the parameters $\tilde a$ and $d$ inherited from 
the original $A$ and on the point $(x,y,z)$ subject to the constraint $x+y+z=1$.
This characteristic equation can be written in the following form:
\begin{equation}
\label{char-equation}
g(\mu;x,y,z):=\mu^3-\frac{\sigma_1(A)}{3}\mu^2+\frac{\sigma_2(A)}{3}(xy+yz+xz)\mu-\sigma_3(A)xyz=0,
\end{equation}
where $\sigma_{1,2,3}(A)$ can be expressed via the already familiar eigenvalues of the matrix $A$:
\begin{equation}
\sigma_1(A)=\Tr\,A=\lambda_1+\lambda_2+\lambda_3,\quad  
\sigma_2(A)=\lambda_1\lambda_2+\lambda_2\lambda_3+\lambda_3\lambda_1,\quad 
\sigma_3(A)=\lambda_1\lambda_2\lambda_3=\det\,A. \nonumber
\end{equation}
The condition \eqref{char-equation} defines the implicit (multivalued) function $\mu$ on $T$.
Choosing $x$ and $y$ as independent variables and substituting $z=1-x-y$, one can compute the total derivatives 
$dg/dx = 0$ and $dg/dy = 0$ and express them via the derivatives
$\mu_x = \partial\mu/\partial x$ and $\mu_y = \partial\mu/\partial y$
of a particular branch of  $\mu$:
\begin{eqnarray}
\label{xy-derivatives}
\frac{d g}{d x} = \frac{\partial g}{\partial \mu} \mu_x + \frac{\partial g}{\partial x} &=&\mu_x \frac{\partial g}{\partial \mu}
+\left(\frac{\sigma_2}{3}\mu-\sigma_3y\right)(z-x)=0\,,\nonumber\\[2mm]
\frac{d g}{d y} = \frac{\partial g}{\partial \mu} \mu_y + \frac{\partial g}{\partial y} &=&\mu_y\frac{\partial g}{\partial \mu}
+\left(\frac{\sigma_2}{3}\mu-\sigma_3x\right)(z-y)=0\,,
\end{eqnarray}
where
\begin{equation}
\frac{\partial g}{\partial \mu} = 3\mu^2-\frac{2\sigma_1}{3}\mu+ \frac{\sigma_2}{3}(xy+xz+yz)\,.\label{dgdmu}
\end{equation}
Vanishing of the derivatives $\mu_x=\mu_y=0$ at a smooth critical point of a given branch implies 
\begin{equation}
\label{main-xy-equations}
\left(\frac{\sigma_2}{3}\mu-\sigma_3 y\right)(z-x)=0\,, \quad 
\left(\frac{\sigma_2}{3}\mu-\sigma_3 x\right)(z-y)=0\,. 
\end{equation}
This system of equations has two sorts of solutions.
\begin{itemize}
	\item 
One is $x=y=z=1/3$, the barycenter of $T$. 
At this point, $A(x,y,z) = A/3$, and the value of $\mu$ is determined by the characteristic equation for $A/3$
whose eigenvalues are $\tilde \lambda_k = \lambda_k/3$.
Let us write this characteristic equation as
\begin{equation}
g(\mu; 1/3, 1/3, 1/3) = (\mu - \tilde \lambda_1)(\mu - \tilde \lambda_2)(\mu - \tilde \lambda_3) = 0\,.\label{gA}
\end{equation}
Then $\partial g/\partial \mu$ at the point $\mu = \tilde \lambda_1$ can be represented as
\begin{equation}
\frac{\partial g}{\partial \mu} = (\tilde \lambda_1 - \tilde \lambda_2)(\tilde \lambda_1 - \tilde \lambda_3)\,.\label{gA2}
\end{equation}
One sees that for a generic matrix $A$, whose eigenvalues do not coincide, $\partial g/\partial \mu \not = 0$.

Differentiating $g=0$ further,
one can observe that the Hessian matrix for $\mu$ at this point has the following form:
\begin{equation}
\mmatrix{\mu_{xx}}{\mu_{xy}}{\mu_{yx}}{\mu_{xx}} = \frac{\mu \sigma_2 - \sigma_3}{3\, \partial g/\partial \mu}
\mmatrix{2}{1}{1}{2}\,.
\end{equation} 
\item
The other solution comes from equating just two variables, for example $x=y\not = z$. 
For a generic $A$ the values of $x,y,z$ forming that solution are well-defined and non-zero:
\begin{equation}
x = y = \frac{\sigma_2 \mu}{3\sigma_3}\,, \quad z= 1-2\frac{\sigma_2 \mu}{3\sigma_3}\,.
\end{equation}
Expressing $\mu=3 \sigma_3 x/\sigma_2$ and substituting it into the characteristic equation \eqref{char-equation},
one obtains the value of $x$:
\begin{equation}
x = \frac{3\sigma_1\sigma_2\sigma_3 - \sigma_2^3}{27\sigma_3^2 - \sigma_2^3}\,,\label{value-of-x}
\end{equation}
 which is, again, well-defined and non-zero for a generic $A$. 
Differentiating $g=0$ again, we obtain at this point
\begin{equation}
0 = \frac{d^2g}{dx^2} = \mu_{xx} \frac{\partial g}{\partial \mu} -2 \left(\frac{\mu \sigma_2}{3} - \sigma_3 y\right) = \mu_{xx} \frac{\partial g}{\partial \mu}\,.
\end{equation}
Direct inspection shows that at this point $\partial g/\partial \mu \not = 0$, as it would attain zero only if $x$ were zero.  
Therefore, $\mu_{xx} = 0$. The same applies to $\mu_{yy}$.
However, the cross derivative is
$\mu_{xy}=\frac{\sigma_3(z-x)}{\partial g/\partial \mu} \not = 0$. 
Thus, the hessian at the considered point is never positive definite, and this critical point is a saddle point, not a minimum.

The same conclusion applies to the other two points, which correspond to $x=z \not = y$ and $y=z \not = x$.
\end{itemize}
We conclude that if $\varphi(x,y,z)$ attains a smooth minimum in the interior of $T$, it can only be at the barycenter $x=y=z=1/3$.
At this point, the three eigenvectors of $A$ given in \eqref{eigenvectors} when divived by $\sqrt{3}$ belong to the fiber over this barycenter
$p^{-1}(1/3,1/3,1/3)$. Therefore, the value $\varphi(1/3,1/3,1/3) = \min \tilde \lambda_k$, $k=1,2,3$.
Thus, if Eqs.~\eqref{A4-BFB-2} are satisfied, $\varphi(1/3,1/3,1/3) + b \ge 0$ and, therefore, by \eqref{chain}, $V_4\ge 0$ on $\mathcal T$.

\subsection{Non-differentiable points of $\varphi$}

The function $\varphi$ defined as the minimum value of the triple cover can, in principle, be non-differentiable at certain points inside $T$.
Let us consider an individual branch of this triple cover. It is represented by the (singular) manifold with boundary, which we denote 
$\mathcal E$ and define as 
$$
\mathcal E=\{(\mu,t)\in \mathbb{R}\times T\,|\,g(\mu,t)=0\},
$$
where the function $g(\mu,t)=g(\mu,x,y,z)$ is defined in \eqref{char-equation}. 
Let $(\mu_0,t_0)\in \mathcal E$ be a point of this manifold. 
The implicit function theorem states that if 
$\frac{\partial g}{\partial \mu}|_{(\mu_0,t_0)}\neq 0$,
then there is an open neighborhood $U$ of $(\mu_0,t_0)$  in $\mathcal E$ such that $U \cong \mathbb{R}$ and
the coordinate $\mu$ of the point $(\mu,t)\in U$ is expressed as a smooth function of the component $t$. 
Therefore, the only obstacle to differentiability of $\mu$ as a function of $t$ at $t_0$
is when $\frac{\partial g}{\partial \mu}|_{(\mu_0,t_0)} = 0$.

This can happen in two distinct situations.
\begin{itemize}
	\item 
\begin{figure}[!h]
	\centering
\begin{tikzpicture}
%\tikz
%{
\draw (-2.5,-1.3) -- (2.5,-1.3);
\draw (-2.5,-1.3) -- (-2.5,1.3);
\draw [black,thick] (0,0) arc (0:-80:2 and 1) (0,0) arc (180:100:2 and 1);
\filldraw [black] (-2.5,0) circle (2pt) (-2.3,0) node[right] {$\mu_0$};
\filldraw [black] (0,-1.3) circle (2pt) (0,-1.5) node[below] {$t_0$};
\filldraw [color=black, fill=white] (0,0) circle (2pt);
\draw (-2.4,1.4) node[right] {$\mu$} (2.4,-1.5) node[below] {$t$};
%}
\end{tikzpicture}
\label{Fig1}
\caption{$\mathcal E$ is smooth at  $(\mu_0,t_0)$ and $\mu(t)$	is not differentiable at $t_0$}
\end{figure}
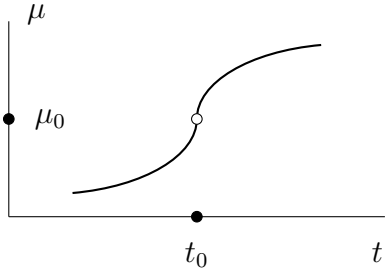
The gradient $\nabla_{\mu,t}g|_{(\mu_0,t_0)}$ whose first component is 
$\frac{\partial g}{\partial \mu}|_{(\mu_0,t_0)}$ is non-zero ``horizontal'' vector.
Then, the variety $\mathcal E$ is smooth at this point, and by virtue of the implicit function theorem, 
$\mu(t)$ takes locally the form of the plot in Fig.~\ref{Fig1}.
Clearly, this point cannot be the extremum of $\varphi$ (actually, one can show that there are no such points in $\mathcal E$ at all). 
	\item
The entire gradient is vanishing at this point. 
This can happen only at the points where at least two branches of the triple cover intersect 
(the case when two branches coincide identically on $T$ can be disregarded for general matrices $A$).
Indeed, following the logic that led us to \eqref{gA2}, we can conclude that $\partial g/\partial \mu = 0$
only when two eigenvalues of $A(x,y,z)$ coincide.
\begin{figure}[!h]
	\centering
	\begin{tikzpicture}
	\draw (-6,-0.8) -- (5.5,-0.8);
%	\draw (-6,-0.8) -- (-6,1.5);
	\draw [black,thick] (-2.5,1) arc (-60:-90:5 and 3) (-2.5,1) arc (120:90:5 and 3) (0,1.4) -- (5,1.4);
	\draw [black,thick] (-2.5,1) arc (60:90:5 and 3) (-2.5,1) arc (-120:-90:5 and 3) (2.5,0.2) arc (60:90:5 and 3) (2.5,0.2) arc (-120:-90:5 and 3);
	\draw [black,thick] (-5,-0.2) -- (0,-0.2) (2.5,0.2) arc (-60:-90:5 and 3) (2.5,0.2) arc (120:90:5 and 3);
% 1-sqrt{3}/2 = 0.134

	\filldraw [color=black, fill=white] (-2.5,1) circle (2pt) (2.5,0.2) circle (2pt);
	\filldraw [black] (2.5,-0.8) circle (2pt);
	\draw (-5.1,1.4) node[left] {$\mu_1$} (-5.1,0.6) node[left] {$\mu_2$} (-5.1,-0.2) node[left] {$\mu_3$} 
	(5.5,-1) node[below] {$t$} (2.5,-1) node[below] {$t_0$};
	\end{tikzpicture}
	\label{Fig2}
	\caption{$\varphi$ is not differentiable at $t_0$, which is the intersection point of two smooth branches}
\end{figure}
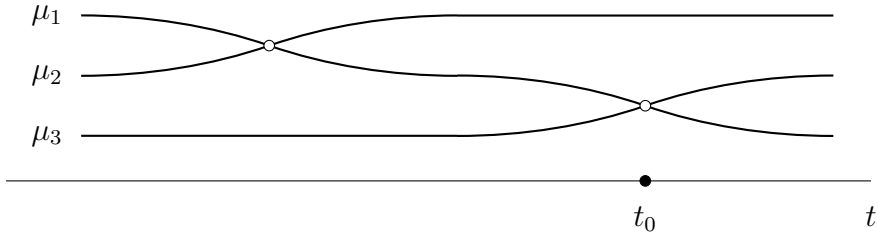
If the intersection takes place in the form as shown in Fig.~\ref{Fig2},
then again it cannot be the global minimum.
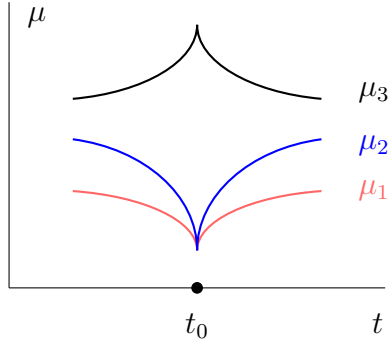
\begin{figure}[!h]
	\centering
	\begin{tikzpicture}
	\draw (-2.5,-1.5) -- (2.5,-1.5);
	\draw (-2.5,-1.5) -- (-2.5,2.3);
	\draw [black,thick] (0,2) arc (0:-80:2 and 1) (0,2) arc (-180:-100:2 and 1);
	\draw [red!60,thick] (0,-1) arc (0:80:2 and 0.8) (0,-1) arc (180:100:2 and 0.8);
	\draw [blue,thick] (0,-1) arc (0:80:2 and 1.5) (0,-1) arc (180:100:2 and 1.5);
	\filldraw [black] (0,-1.5) circle (2pt) (0,-1.7) node[below] {$t_0$};
	\draw (-2.4,2.1) node[right] {$\mu$} (2.4,-1.7) node[below] {$t$} (2, 1.1) node[right] {$\mu_3$};
	\draw[blue] (2, 0.4) node[right] {$\mu_2$};
	\draw[red!60] (2, -0.2) node[right] {$\mu_1$};
	\end{tikzpicture}
	\label{Fig3}
	\caption{$\varphi$ is not differentiable at the common cusp point of $\mu_{1,2}$}
\end{figure}
The only arrangement for such a point to be the global minimum is when two branches not only coincide but also both have downward cusps
at this point, Fig.~\ref{Fig3}. 
However, since $\mu_1+\mu_2+\mu_3 = \sigma_1/3 =\mathrm{const}$, it implies that the \textit{third}
branch must possess a compensating upward cusp. But if it has a cusp, it must, at this point, meet yet another eigenvalue.
Since we have only three eigenvalues at our disposal, we conclude that this arrangement is  impossible. 
\end{itemize}
These geometric arguments can be recast in the strict differential language which confirms the conclusion
that $\varphi$ cannot attain the global minimum at the interior points of $T$ where it is non-differentiable.

\subsection{The function $\varphi$ on the boundary of $T$}

Let us now restrict the function $\varphi(x,y,z)$ to the boundary of $T$.
For definiteness, we consider the side $z=0$. Then, we deal with one-variable problem, with $x$
being independent variable and $y=1-x$.
The characteristic equation \eqref{char-equation} takes the form
\begin{equation}
\mu^3-\frac{\sigma_1}{3}\mu^2+\frac{\sigma_2}{3}x(1-x)\mu=\mu\left[\mu^2-\frac{\sigma_1}{3}\mu+\frac{\sigma_2}{3}x(1-x)\right]=0.
\label{edge-equation}
\end{equation}
One of the branches of our multivalued functional solution of the above equation is $\mu = 0$, which is constant  along the boundary of $T$.
In fact, as can be seen from the original characteristic equation \eqref{char-equation},
for a generic matrix $A$, the spectrum of $A(x,y,z)$ contains zero only at the boundary $xyz=0$.

Below, by ``minimum'' we always mean ``the absolute minimum''. 
The case when $\mu=0$ delivers the minimum of $\varphi$ on $T$ will be covered in the next subsection; 
now we assume that if the  minimum of 
$\varphi$ is attained at the boundary, it is nonzero. Moreover, we
assume that locally around the  point of the minimum $\varphi$ is equal to the minimum
of the  two branches determined by
the remaining quadratic factor in  \eqref{edge-equation}. 
%As for non-zero values of $\mu$, they come from the remaining quadratic equation.

Applying the differential analysis of critical and non-differentiable points just as above, we see that
the minimum of $\varphi$ can be attained either at the side barycenter $x=y=1/2$
or at the vertices.
In the case when the minimum is attained at the side barycenter, it is equal to
$$
\varphi(1/2,1/2,0)=\min \,Spec\, A(1/2,1/2,0)=\min\,\{0, (\tilde a-|d|)/2,(\tilde a+|d|)/2\}=(\tilde a-|d|)/2.
$$
Thus, if Eqs.~\eqref{A4-BFB-1} are satisfied, 
$\varphi(1/2,1/2,0) + b \ge 0$. If the minimum is attained at $(1,0,0)$, it is equal to
$$
\varphi(1,0,0)=\min \,Spec\, A(1,0,0)=\min\,\{0, \tilde a\}=\tilde a,
$$
and again by Eqs.~\eqref{A4-BFB-1} we have 
$\varphi(1,0,0) + b \ge 0$. In either case we get by \eqref{chain} that $V_4 \ge 0$ on $\mathcal T$.

\subsection{The case $\min\,\varphi = 0$}

The assumption $\underset{(x,y,z)\in T}{\min}\,\varphi(x,y,z)=0$ implies that the spectrum of $A(x,y,z)$ for all $(x,y,z) \in T$ is non-negative, in particular the hermitean form $h(\vec{z})$ with the matrix 
$A=3A(1/3,1/3,1/3)$ is non-negative. 

We can equivalently reformulate inequalities \eqref{A4-BFB-2} 
in terms of the eigenvalues $\lambda_k$ of $A$ as
\begin{equation}
\frac{1}{3}\left(\underset{k=1,2,3}{\min}\ \,\lambda_k\right)+b\ge 0.
\end{equation}
The form
$h(\vec{z})$ attains its minimum on the sphere $||\vec z||^2=1/3$ at an 
eigendirection which, by \eqref{eigenvectors}, belongs to the fiber
 $p^{-1}(1/3,1/3,1/3) \subset \mathcal T$. The entire fiber $p^{-1}(1/3,1/3,1/3)$ is contained in the sphere $||\vec z||^2=1/3$. It is an easy check that this fiber consists
of the shortest vectors in $\mathcal T$, that is, all vectors $\vec{z}\in \mathcal T$ satisfy
$||\vec{z}||^2\geqslant 1/3$. Consequently, for our non-negative $h(\vec{z})$ we have 
$$
V_4(\vec{z})=h(\vec{z})+b=
3||\vec z||^2 h\left(\frac{\vec z}{\sqrt{3}||\vec z||}\right)+b\ge  h\left(\frac{\vec z}{\sqrt{3}||\vec z||}\right)+b
\ge$$
\begin{equation}
\ge  \underset{\vec{z}_0\in p^{-1}(1/3,1/3,1/3)}{\min}\,h(\vec{z}_0)+b=\frac{1}{3}\left(\underset{k=1,2,3}{\min}\ \,\lambda_k\right)+b\ge 0,
\end{equation}
for all $\vec{z}\in \mathcal T$. 
Non-negativity of $h(\vec z)$ is used in the first inequality in this chain.

\section{Discussion and conclusions}

\subsection{The synthesis of the arguments}

Let us summarize the arguments and calculations of the previous sections.
\begin{itemize}
\item Necessity. Imposing conditions \eqref{A4-BFB-1} and \eqref{A4-BFB-2} is necessary
	to ensure $V_4 \ge 0$ because we explicitly found the rays in $\CC^3$ where
	$V_4$ is proportional to each of the expressions involved in these inequalities.
	Thus, the most laborious task is to prove the sufficiency of \eqref{A4-BFB-1} and \eqref{A4-BFB-2} for $V_4 \ge 0$.
\item
	Knowing that $V_4$ is a homogeneous function, we restricted it to the compact space ${\mathcal T}$ defined by 
	$|z_1|+|z_2|+|z_3| = 1$. Instead of analyzing $V_4$ on ${\mathcal T}$, we introduced an auxiliary function
	$\varphi$ defined on the triangular region $T \subset \RR^3$ of points $(x,y,z)$ satisfying $x,y,z \ge 0$ and $x+y+z=1$.
	This function has the property that $V_4 \ge \varphi + b$ at any point where $V_4$ has an extremum.
	In other words, the global minimum of $V_4$ in ${\mathcal T}$ is bounded from below by the global minimum of $\varphi+b$ on $T$.
	Thus, we switched our attention to the minima of $\varphi$.
\item
	Since $\varphi$ depends on two independent variables, it is easier to analyze than $V_4$ in ${\mathcal T}$.
	We found that non-zero minima of $\varphi$ can be located only at the barycenters of $T$ \eqref{barycenters}.
	If the conditions \eqref{A4-BFB-1} and \eqref{A4-BFB-2} are imposed, then $\varphi + b \ge 0$ at the barycenters
	and, consequently, $V_4 \ge 0$ everywhere in ${\mathcal T}$.
\end{itemize}

\subsection{Physics reformulation}

Three-Higgs-doublet models make use of three Higgs doublets $\phi_i$, $i = 1,2,3$, each $\phi_i \in \CC^2$. 
The Higgs potential of the $A_4$-invariant 3HDM 
\cite{Ma:2001dn,Ishimori:2010au,Toorop:2010ex,Ivanov:2012fp}
is traditionally written as
\begin{eqnarray}
V &=& \mu^2(\phi_1^\dagger\phi_1 + \phi_2^\dagger\phi_2 + \phi_3^\dagger\phi_3) + 
\lambda_{1}(\phi_1^\dagger\phi_1 + \phi_2^\dagger\phi_2 + \phi_3^\dagger\phi_3)^2 \nonumber\\[2mm]
&& + \lambda_{3}\left[(\phi_1^\dagger\phi_1)(\phi_2^\dagger\phi_2) + (\phi_1^\dagger\phi_1)(\phi_3^\dagger\phi_3) + (\phi_2^\dagger\phi_2)(\phi_3^\dagger\phi_3)\right] + \lambda_{4} (|\phi_1^\dagger\phi_2|^2 + |\phi_2^\dagger\phi_3|^2 + |\phi_3^\dagger\phi_1|^2)\nonumber\\[2mm]
&& + \frac{\lambda_{5} }{2} \left\{e^{i\epsilon}\left[(\phi_1^\dagger\phi_2)^2 + (\phi_2^\dagger\phi_3)^2 + 
(\phi_3^\dagger\phi_1)^2 \right] + H.c.\right\}\,.\label{VA4-1}
\end{eqnarray}
One usually assumes that the minimum of this potential corresponds to a neutral vacuum, with the vacuum expectation values 
of the three doublets being proportional to each other as vectors in $\CC^2$.
This allows one to focus attention only to the ``neutral subspace'' of the space of three doublets. 
Technically, it is done by replacing $\phi_i$ with complex numbers, so that the terms with $\lambda_3$ and $\lambda_4$
can be combined in a single term with the coefficient $\lambda_3 + \lambda_4$.
Finally, by defining $z_i = \phi_i^2$, one rewrites the quartic part of the potential 
as a quadratic form in terms of $z_i$ and their conjugates \eqref{V4-1}.
The relation between the coefficients is
\begin{equation}
a = \lambda_1\,, \quad 2b = 2\lambda_1 + \lambda_3 + \lambda_4\,, \quad 2d = \lambda_5 e^{i\epsilon}\,.
\end{equation}
The inequalities \eqref{A4-BFB-1} and \eqref{A4-BFB-2} then become
\begin{equation}
\lambda_1 \ge 0\,, \quad 4\lambda_1 + \lambda_3 + \lambda_4 - \lambda_5 \ge 0\,, \quad
3\lambda_1 + \lambda_3 + \lambda_4 + \lambda_5 \cos\left(\epsilon +\frac{2\pi k}{3}\right) \ge 0\,,\label{BFB-physics}
\end{equation}
just as stated in \cite{Ivanov:2020jra}.
In this work we proved that these conditions are necessary and sufficient BFB conditions in the neutral Higgs subspace.

In the light of the results of \cite{Faro:2019vcd,Ivanov:2020jra}, a word of caution is in order.
As demonstrated in \cite{Faro:2019vcd}, a multi-Higgs potential can be bounded in the neutral subspace but unbounded from below
along some charge-breaking directions, even if it has a valid neutral (local) minimum, with all the charged Higgs masses squared positive.
Since we have not analyzed the full space of Higgs doublets, we still cannot claim the conditions \eqref{BFB-physics}
are sufficient in all situations with neutral vacuum.
In particular, they could easily become insufficient in softly broken $A_4$ symmetric model. 
To close this gap, one would need to explore the quartic potential in the entire orbit space, including the charge-breaking directions,
similarly to what was done in \cite{Ivanov:2020jra} for the $S_4$-symmetric 3HDM.

Finally, we remark that the auxiliary function method described in this paper may also be suitable
for the less symmetric versions of 3HDM such as Weinberg's and Branco's models with the symmetry group $\Z_2 \times \Z_2$.
However, this passage is not trivial and requires further work.

\subsection{Conclusions}

In summary, we proved in this paper that the inequalities \eqref{BFB-physics} conjectured in \cite{Ivanov:2020jra}
indeed represent the necessary and sufficient bounded-from-below conditions for the $A_4$ symmetric 3HDM,
at least within the neutral Higgs space.
We could achieve this result with the aid of a novel technique, in which we used an auxiliary function $\varphi$
defined in \eqref{min-function} and related to the Higgs potential.
We believe that this technique can be extended to the more challenging case of the original Weinberg's 3HDM model,
where the exact BFB conditions are still unknown.

\subsection*{Acknowledgements}
The work was initiated in July-August 2020 during the Workshop of Mathematical center in Akademgorodok, held under agreement No. 075-15-2019-1675 with the Ministry of Science and Higher Education of the Russian Federation.
We are grateful to the organizers of this workshop for the stimulating atmosphere. 
The initial steps towards the proof reported here,
including numerical verification of the validity of the conjectured inequalities, were done
in collaboration with Valery Churkin, Daria Lytkina,  
Ivan Antipov, Yury Efremenko, Igor Novikov, Arseny Sadovnikov, Kirill Starodubets, and Alexandr Chumakov. 
We also thank Prof. Valery Churkin for suggesting the idea of using the auxiliary function $\varphi$.

\appendix

\section{Generic and non-generic points in matrix space}\label{appendix-generic-matrices}
\label{Generic-section}
\subsection{The parameter space for functions $V_4$}
The functions $V_4$ are encoded by the tuples $(A,b)$ of their 
coefficients, $A$ is the matrix of the hermitean part, $b$ is the coefficient
of the remaining non-hermitean part. Given the specific form of the matrices $A$,  the space of all such tuples can be identified with 
$\mathbb{R}^4$. We introduce a particular subset $\mathcal H \subset \mathbb{R}^4$,
\begin{equation}
\mathcal H=\{(A,b) \in \mathbb{R}^4\,|\,(A,b) \mbox{ satisfies Eqs.~\eqref{A4-BFB-1}, \eqref{A4-BFB-2} }\,\}.
\end{equation}
Let $V_{4,p}$ denote the function $V_4$ with the datum $p=(A,b)$. 
Thus Theorem \ref{Main-theorem} states that $V_{4,p}\ge 0$ on $\mathcal T$ if and only if $p \in \mathcal H$. The subset $\mathcal H$ admits 
a simple and useful geometric description, formulated in the following lemma. 
\begin{lemma}
\label{cone-lemma}
The subset $\mathcal H$ is a convex cone in $\mathbb{R}^4$,
that is, for every $p,q\in \mathcal H$, $\lambda,\mu \in \mathbb{R}_{\ge 0}$ we have $\lambda p+ \mu q \in \mathcal H$. This cone has non-empty
interior, which is dense in $\mathcal H$. 
\end{lemma}
\begin{proof}
Let $p,q$ be any points in $\mathcal H$. 
By Proposition \ref{barycenter-proposition}, the functions
$V_{4,p}$ and $V_{4,q}$, determined by the tuples $p$ and $q$
are non-negative on the fibers over barycenters of $T$. Then,
for any $\lambda,\mu \in \mathbb{R}_{\ge 0}$ the function
$V_{4,\lambda p +\mu q}=\lambda V_{4,p}+\mu V_{4,q}$
is also non-negative over the barycenters,  so that, again by
Proposition  \ref{barycenter-proposition}, $\lambda p +\mu q \in \mathcal H$. 

The statement that $\mathcal H$ has non-empty interior means that
$\mathcal H$ contains a non-empty open subset of $\mathbb{R}^4$. Let us choose
$p_0=(A,b)\in \mathbb{R}^4$ such that $V_{4,p_0}$ satisfies strict versions
of inequalities  \eqref{A4-BFB-1}, \eqref{A4-BFB-2}. As $V_{4,p}$
depends continuously on $p$, one may vary $p$ in a small enough open neighborhood $U \subset \mathbb{R}^4$ of $p_0$, 
so as to keep $V_{4,p}$ satisfy \eqref{A4-BFB-1}, \eqref{A4-BFB-2} for all $p \in U$. Then $U\subset \mathcal H$.

The statement about density of the interior of $\mathcal H$ now follows from
non-emptiness of the interior and convexity of $\mathcal H$ 
(in short, we can always move from a point on the boundary of $\mathcal H$ to an interior point of $\mathcal H$ along a segment in 
$\mathcal H$).
\end{proof}

\subsection{Generic and non-generic points}
Assume now that we have a finite collection $f_1,\dots,f_n$ of 
not identically zero polynomial or real-analytic functions on $\mathbb{R}^4$.
Each of the equations $f_i=0$ determines a real-analytic subvariety (perhaps, of codimension $> 1$ or even empty)
 $P_i\subset \mathbb{R}^4$. The points of $\mathbb{R}^4\setminus \underset{i=1}{\overset{n}{\bigcup}}P_i$ are called {\it generic} (with respect to the collection $f_1,\dots,f_n$). The points of the union $\underset{i=1}{\overset{n}{\bigcup}}P_i$ are called {\it non-generic}
(with respect to $f_1,\dots,f_n$).

The following proposition allows us to reduce proving Theorem \ref{Main-theorem} for $V_{4,p}$ with $p\in \mathcal H$ to proving it
only for generic $p \in \mathcal H$.
\begin{proposition} If $V_{4,p}\ge 0$ on $\mathcal T$ for every generic 
$p\in \mathcal H$, then $V_{4,p} \ge 0$ for every $p \in \mathcal H$. 
\end{proposition}
\begin{proof}
Given a $p_0\in \mathcal H$, by density of the interior
and convexity of $\mathcal H$ (see  Lemma \ref{cone-lemma}) 
we can choose $p_1$ in the interior of $\mathcal H$  
 determining a segment $[p_0,p_1]\subset \mathcal H$, such that 
all points of $(p_0,p_1]$ are generic. 
Then by our assumption 
$V_{4,q}\ge 0$ on $\mathcal T$ for all $q\in (p_0,p_1]$,
so that $$V_{4,p_0}=\underset{\begin{array}{c}q\to p_0\\ q\in(p_0,p_1] \end{array}}{\lim}\, V_{4,q}\ge 0$$
on $\mathcal T$. 
\end{proof}

\subsection{The collection of functions on $\mathbb{R}^4$}
\label{Collection-subsection}
We define our collection as the union of two parts: the part of real-analytic functions of the off-diagonal entry $d$ of $A$
\begin{equation}|d|, \quad arg\, d-\frac{\pi k}{3},\quad k=0,1,2,3,4,5,
\end{equation}
and the part of polynomial functions
\begin{equation}\sigma_2(A), \,\,\sigma_3(A), \,\,
 3\sigma_1(A) \sigma_3(A)-\sigma_2^2(A),\,\,
27\sigma_3^2(A)-\sigma_2^3(A)\end{equation}
where $\det (t{\mathbf 1}-A)=t^3-\sigma_1(A)t^2+\sigma_2(A)t-\sigma_3(A)$ 
(in fact, these functions depend only on $A$-component of $(A,b) \in \mathbb{R}^4$). 
The check that the listed functions are not identically zero on $\mathcal H$ is straightforward but we omit it here. 

Looking at formulas \eqref{spectrum} we note that two eigenvalues
of $A$ coincide only when $|d|=0$ or one of $arg\,d-\frac{\pi k}{3}$ vanishes, 
for that $A$.  
Thus, the union of subvarieties $P_k \subset \mathbb{R}^4$ determined by 
respective equations $|d|=0$, $arg\,d-\frac{\pi k}{3}=0$ contains all $(A,b) \in \mathbb{R}^4$
with the matrix $A$ having a multiple eigenvalue, so the generic, 
with respect to that part of our collection, points $(A,b)$ are those with
$A$ having simple spectrum.

The generic, with respect to the second part of our collection, points $(A,b)$ are those, for which the expressions for the critical points of $\varphi$, appearing in subsection \ref{smooth-interior-minima-section}, and analogous expressions in further subsections are well-defined.

\end{document}